\documentclass[showpacs,amsmath,amssymb,twocolumn,pra]{revtex4-2}

\usepackage[dvips]{graphicx} 
\usepackage{amsfonts}
\usepackage{amssymb}
\usepackage{amscd}
\usepackage{amsmath}    
\usepackage{enumerate}
\usepackage{epsfig}
\usepackage{subfigure}
\usepackage{bm}
\usepackage{xcolor}
\usepackage{amsthm}
\usepackage{framed}
\usepackage{multirow}
\usepackage{mathrsfs,amssymb}
\usepackage{latexsym} 
\usepackage{framed}
\usepackage{physics}
\usepackage{epstopdf}
\usepackage[skins]{tcolorbox}
\usepackage{hhline}

\newtcolorbox[auto counter]{tbox}[2][]{%
    enhanced, float=hbt, drop fuzzy shadow southeast,
    colback=white!5!white, colframe=white!50!black,
    width= .97\columnwidth,sharp corners, boxrule=0.8pt,
    title={Box \thetcbcounter: #2}, #1
}

\newtheorem{theorem}{Theorem}
\newtheorem{lemma}{Lemma}
\newtheorem{corollary}{Corollary}

\newtheorem{definition}{Definition}

\begin{document}
\title{Quantum Network: Security Assessment and Key Management}

\author{Hongyi Zhou}
\author{Kefan Lv}
\author{Longbo Huang}
\email{longbohuang@tsinghua.edu.cn}
\author{Xiongfeng Ma}
\email{xma@tsinghua.edu.cn}
\affiliation{Institute for Interdisciplinary Information Sciences, Tsinghua University, Beijing 100084, China}

\begin{abstract}
As an extension of quantum key distribution, secure communication among multiple users is an essential task in a quantum network. When the quantum network structure becomes complicated with a large number of users, it is important to investigate network issues, including security, key management, latency, reliability, scalability, and cost. In this work, we utilize the classical and graph theories to address two critical issues in a quantum network: security and key management. First, we design a communication scheme with the highest security level that trusts a minimum number of intermediate nodes. Second, when the quantum key is a limited resource, we design key management and data scheduling schemes to optimize the utility of data transmission. Our results can be directly applied to the current metropolitan and free-space quantum network implementations and potentially be a standard approach for future quantum network designs.
\end{abstract}

\maketitle

\section{Introduction}\label{sec:introduction}
As a principal part of quantum cryptography, quantum key distribution (QKD) allows remote communication parties to share identical and private keys for encryption and decryption \cite{Bennett1984Quantum,Ekert1991Quantum}, whose information-theoretical security is guaranteed by the fundamental principles of quantum mechanics \cite{Lo1999Unconditional,Shor2000Simple}. The practical implementation of QKD has a booming development since the beginning of this century. For the most popularly applied photon source --- highly attenuated weak coherent state light, the decoy-state method \cite{Hwang2003Decoy,lo2005decoy,Wang2005Decoy} addresses security issues caused by the information leakage of multi-photon components. Since then, many long-distance QKD experiments have been demonstrated around the world \cite{Zhao2006ExpDecoy,PhysRevLett.98.010503,PhysRevLett.98.010504,PhysRevLett.98.010505, Liu:10,boaron2018secure}. In the meantime, the measurement-device-independent quantum key distribution (MDI-QKD) protocol has been proposed to address the detection loophole problems \cite{lo2012mdiqkd}, which has been demonstrated both in the lab \cite{Liu2013ExpMDI,Tang2013experimental,Silva2013DemoPolMDIQKD,tang2014measurement, yin2016measurement} and in the field \cite{Rubenok2013MDIQKDFielfTest,Tang2015MDIField}. Recently, theoretical development on MDI-QKD shows that one can further double the secure communication distance \cite{Ma2018Phase,Lin2018simple}. All these developments suggest that point-to-point QKD over hundreds of kilometers is ready for real-life implementation.

The initial proposal of QKD deals with a two-user communication scenario. In practice, one needs to extend point-to-point links to a network. There are in principle two types of quantum networks according to the way of connection \cite{salvail2010security}: one is optically switched quantum networks realized by classical optical functions; the other is repeater-based networks. The quantum-repeater-based networks are fully quantum networks enabling multi-partite entanglement distribution \cite{epping2017multi,takeoka2019multipartite,das2019universal}. While in practical implementations, restricted by the current technology, it is the classical repeaters, i.e., trusted intermediate nodes together with optical switches that composites different topological network structures. To this day, there have been a number of experimental demonstrations on the field test of quantum networks. Several testing implementations of quantum networks have been realized in the China, Europe, Japan, and USA \cite{elliott2005darpa,peev2009secoqc,chen2009field,sasaki2011field,wang2014field}. Today, the topological structures of QKD networks have become more complex than the early ones \cite{elliott2005darpa,peev2009secoqc,sasaki2011field}, such as the mesh structure in 46-node Hefei network, and the star-type structure in MDI-QKD network \cite{PhysRevX.6.011024}. Besides these fiber-based quantum networks, a satellite-relayed quantum network has been realized recently \cite{liao2018satellite}, in which a secret key was exchanged between intercontinental communication partners.
In the mean time, researchers explore the feasibility of hybridizing discrete variable schemes with continuous variable ones \cite{PhysRevA.94.042340,elmabrok2018quantum} in a quantum network and integrating QKD into classical networks, such as utilizing wavelength division multiplexing technique \cite{chapuran2009optical,qi2010feasibility}.

The ultimate goal of quantum communication is to realize large scale quantum networks. There are a few major challenges that a quantum network faces, including (i) designing the proper topological structure, (ii) assessing the security levels, and (iii) managing secure keys. Recently, an engineering framework of a scalable multi-site quantum network has been established \cite{tysowski2018engineering}, in which a QKD system is divided into multiple layers: host layer, key management layer, QKD network layer, and quantum link layer. The core of designing and operating a quantum network lies in the key management layer and QKD network layer, where the issues of security, key management \cite{pattaranantakul2012secure}, data routing \cite{yang2017qkd} and stability should be dealt with to realize the optimal transmission performance at a low cost.

In a quantum network, communication between two users, Alice and Bob, are often relayed by intermediate nodes. These nodes can be divided into two types, trusted nodes and untrusted nodes, depending on whether or not the security of communication relies on the security of the nodes, respectively. A trusted node executes full QKD process with adjacent nodes and announces the parities of the two key bit strings such that end users can share secret keys. For example, Alice and Bob establish keys, denoted by $k_a$ and $k_b$, with an intermediate trusted node. The trusted node announces $k_a\oplus k_b$, and eventually, Alice and Bob share a key $k_a$. Whereas, an untrusted node can be as simple as an optical switch, or it can be an untrusted measurement site used in MDI-QKD schemes. In a simple network structure, such as the MDI-QKD network \cite{PhysRevX.6.011024}, users can communicate without trusting any intermediate nodes. Since the security of communication does not depend on untrusted nodes, we can only consider the trusted nodes in the following discussions on security assessment. In this paper, we assume all users in the network are connected with insecure classical communication channels, which are treated as a free resource. We focus on the case where quantum keys are consumed for private communication.

Security is a crucial issue in a quantum network, which may be compromised if an adversary, Eve, can manipulate or crack intermediate trusted nodes. In reality, it is important to evaluate security if Eve can at least compromise one of the nodes. In the extreme case where Eve can hack all the intermediate nodes, no secure communication can be established. Thus, we first consider the interesting problem of security assessment when a certain amount of the nodes are compromised. In particular, for a general network structure, we find an optimal communication scheme using the graph representation of a quantum network \cite{PhysRevA.95.012304,pirandola2019end}. With the highest security level, the communication is secure unless the nodes that eavesdropper compromises form a cut in the graph of the quantum network, which is the ultimate solution of an efficient attack strategy by the adversary.

Another issue addressed in this work is the data routing and key management in the network. In private key systems such as QKD, the encoding process may consumes keys with the same length as the message. With the current quantum technology, the key generation speed (10 Mbps) is far below the speed of classical data transmission (1 Tbps). Thus, the key is a limited resource in a network for most communication tasks. In a network with multiple communication tasks, the lack of key management will lead to instability and inefficiency. To address this issue and optimize network management, we adopt techniques addressing similar problems in the classical network research. Specifically, we formulate the problem of QKD-based network communication as a flow scheduling problem in resource-constrained networks, such as processing networks \cite{Jiang2010Scheduling,Dai2008Asymptotic} and energy-harvesting networks \cite{Longbo2013Utility,Chen2014A}. In this formulation, each data transmission consumes supporting resources (in our case, quantum key bits), and the network operator needs to jointly optimize resource usage, data routing, and scheduling. Then, to solve the problem of key-constrained data transmission, we adopt the Lyapunov network optimization technique \cite{Georgiadis2006Resource} and design a key management and data scheduling algorithm that has low-implementation complexity. We also rigorously show that our algorithm achieves near-optimal performance in terms of data transmission utility.


%
%
%
%
\section{Security assessment} \label{Sec:Security}
The security in a quantum network lies in two aspects: quantum channel and intermediate nodes. The former has been well studied in the security analysis of QKD; while the latter is a new problem emerging in quantum networks. Trusted nodes can extend communication distances, keeping a relatively high key rate at the mean time. At a cost, the security of communication can be compromised by the trustworthiness of intermediate relay nodes. In practice, an important task is to design a key exchange procedure, so that it can tolerate the maximal number of compromised nodes. We define the tolerance of the compromised nodes as security level. Our target is to find a communication scheme that can tolerate as many compromised node as possible, i.e., with the highest security level.

In this section, we first present our network model. Then, we consider several simple communication schemes and provide the corresponding attack strategies. After that, we propose the strongest attacks that can hack all possible communication schemes. We find a communication strategy with the highest security level that is secure unless Eve performs the strongest attack.


\subsection{Network model}
In our model, we consider two networks. One is a classical network for data transmission, and the other is a quantum network for key generation. When two users Alice and Bob want to establish a communication, they first distribute secret keys via the quantum network and then encrypt the message and transmit it by the classical network. A quantum network can be represented by a graph $G=(\mathcal{N},\mathcal{L})$, where $\mathcal{N}$ and $\mathcal{L}$ are the sets of vertices and edges, respectively. Here a vertex $c\in\mathcal{N}$ represents a basic unit in a quantum network, which can be a node or a QKD sub-network whose internal structure is unrelated to the security assessment. An edge in the graph represents a QKD link used to distribute secure key strings between connected nodes. While a classical network is represented by another graph $G^\prime=(\mathcal{N},\mathcal{L^\prime})$ sharing the same vertices with the quantum one. Edges in the graph $G^\prime$ represent classical links, which may be different from the QKD links. In this section, we assume the classical links can be freely used and hence neglect the classical communication efficiency in the following discussions.

We focus on the security of nodes and trust the security of QKD links, i.e., we assume the QKD process has been completed and secret keys have been generated. For example, in this model, the untrusted measurement site in MDI-QKD is merged into the QKD links in as an edge in the graph $G$ . We have the following assumptions for the adversary Eve.  a) Eve has access to all classical channels. b) Eve has no information about the quantum key of an edge if she does not compromise either of the connected nodes.  c) Eve learns everything of the quantum key if she compromises at least one of the connected nodes. We present a toy quantum network $G$ in Fig.~\ref{fig:strategy1}, where Alice and Bob are two communicating parties. There are $5$ intermediate nodes between them denoted by $c_1\sim c_5$, and $9$ edges $k_1\sim k_9$, each representing a quantum channel or quantum key strings generated between the connected nodes.
\begin{figure}[hbt]
\centering
\resizebox{8cm}{!}{\includegraphics{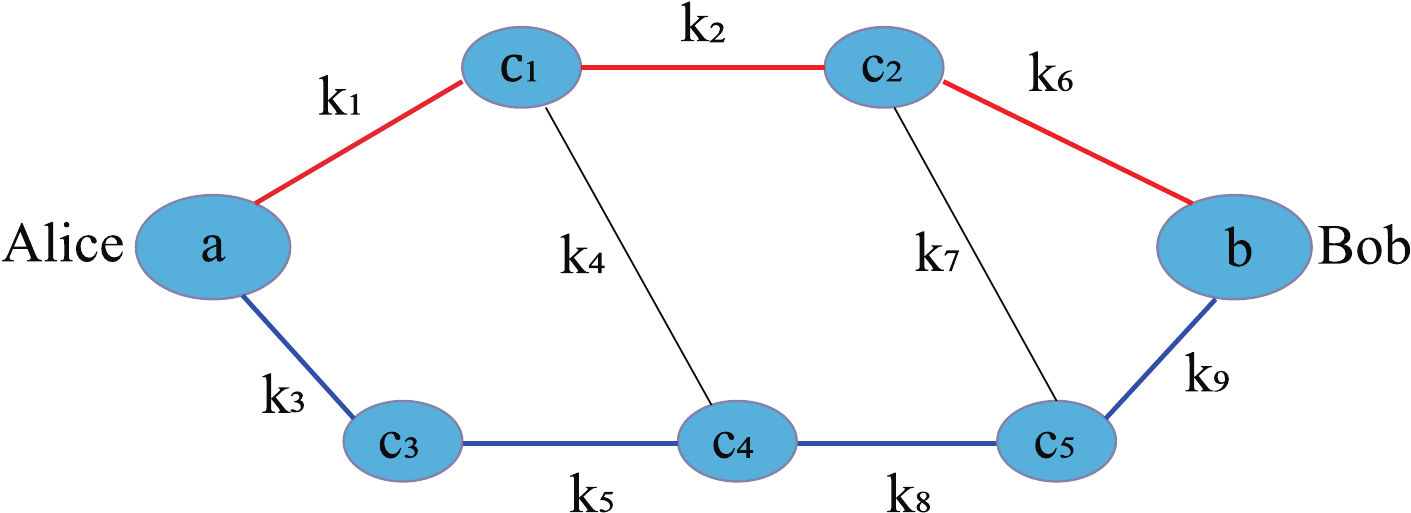}}
\caption{Graph representation of a quantum network. Alice and Bob are the two communicating parties denoted by $a$ and $b$, respectively. There are $5$ intermediate nodes between them, denoted by $c_1\sim c_5$. Each edge represents a quantum channel or quantum key strings generated between the connected nodes, denoted by $k_1\sim k_9$.}
\label{fig:strategy1}
\end{figure}

Here, we use the path concept in the graph theory to describe the sequence of nodes used in message transmission. We only consider simple paths here since any loop of the message transmission is useless in a network. Take the red line in Figure \ref{fig:strategy1} as an example. Once Alice and Bob pick a path, all intermediate nodes are fixed, $a\rightarrow c_1\rightarrow c_2 \rightarrow b$. There are two means for key sharing. One is that Alice and Bob ask all the intermediate nodes ($c_1$ and $c_2$ in this example) to announce the parities from the exclusive-or (XOR) operations on the key bit strings with their two neighbors in the path. In the example, $c_1$ announces $k_1\oplus k_2$ and $c_2$ announces $k_2\oplus k_6$. Then Bob can share the same key with Alice by calculating the parities. Alternatively, $c_2$ will first calculate Alice’s key with the parity announced by $c_1$, encrypt Alice’s key with $k_6$, and sends it to Bob. In this method, Alice’s key can be regarded as a message. In theory, both ways of private communication are equally secure. The first method uses fewer encryptions and decryptions. Also, the intermediate nodes do not obtain the final key directly. However, in practice, the first method will lead to low communication efficiency since all the intermediate nodes should publicly announce their parity bits for the key exchange of end-users. The second method is similar to the data transmission in a realistic network, i.e., a node will receive an encrypted message from one of its neighbors and a request to convey the message to another neighbor privately. The difference between the two methods will vanish if we neglect classical communication efficiency. Therefore, we do not distinguish these two methods in the following discussions.

\subsection{Multi-path communication scheme and strongest attack}
Let us begin with a simple case with only one single-line path, which is represented by the red line $a\rightarrow c_1\rightarrow c_2 \rightarrow b$ in Figure~\ref{fig:strategy1}. Alice sends the message to her neighbour relay node encrypted with the quantum key. The message is decrypted and re-encrypted by intermediate nodes and finally received by Bob. This is a strategy that consumes the least amount of keys. In terms of security, this scheme can be weak because once Eve cracked any node on the path, she gets the message.

In order to strengthen the single-path strategy, one can introduce an additional disjoint path, the blue line $a\rightarrow c_3\rightarrow c_4 \rightarrow c_5 \rightarrow b$ shown in Figure \ref{fig:strategy1}, to defend the single-point eavesdropping attack. The second path is used to transmit another independent random bit string. The final message is the XOR result of the two strings transmitted via the two paths. The details of the communication scheme is shown in Table \ref{tab:2path}.

\begin{table}
\caption{Two-path communication scheme. The two paths are the red line $a\rightarrow c_1\rightarrow c_2 \rightarrow b$ and the blue line $a\rightarrow c_3\rightarrow c_4 \rightarrow c_5 \rightarrow b$ shown in Figure \ref{fig:strategy1}.}
\label{tab:2path}
\begin{framed}
\begin{enumerate}
\item
Alice holds the message $x$ and generates a random bit string $y$ locally with the same length $|x|=|y|$.
\item
Alice sends the message $x \oplus y$ to her neighbor node $c_1$ and $y$ to $c_3$, encrypted by quantum keys.
\item
The node $c_1$ sends $x \oplus y$ through the red path to $c_2$, and eventually to Bob. Similarly, $c_3$ sends $y$ through the blue path to $c_4$, to $c_5$, and eventually to Bob.
\item
Finally, Bob receives $y$ and $x \oplus y$ from the two different paths. He obtains the message $x$ by applying an XOR operation, $x=x \oplus y\oplus y$.
\end{enumerate}
\end{framed}
\end{table}

Suppose Eve can only successfully hack one of the intermediate nodes, she can only learn $y$ or $x \oplus y$, and hence the transmission of $x$ is still secure. In fact, if Eve can only hack the nodes in one of the two paths (red or blue), the transmission is secure. Only when Eve can hack nodes from both paths, say $c_2$ and $c_3$, she can eavesdrop the message. Obviously, the two-path scheme is securer than the single-path one in practice.

Generally, we can increase the number of paths to increase communication security. The two-path scheme shown in Box \ref{tab:2path} can be generalized to a multi-path scheme. In an $n$-path scheme, Alice sends $x \oplus y_1 \oplus y_2 \cdots \oplus y_{n-1}, y_1, y_2, ..., y_{n-1}$ by $n$ paths. Then it follows step 2, 3 and 4 in Table~\ref{tab:2path}. Finally, Bob can recover the message $x$. Sometimes, adding a path may not increase security. For example, in the aforementioned two-path strategy, adding the third path, $a\rightarrow c_1\rightarrow c_4\rightarrow c_5 \rightarrow b$, cannot enhance security since any hacking strategy that can successfully break the security of the two-path scheme will also break this three-path scheme. Our target is to design a robust communication scheme against as many compromised nodes as possible.

Now, we model communication schemes and define the security levels regarding a quantum network formally. In a quantum network represented by a graph $G=(\mathcal{N},\mathcal{L})$, we denote $\mathcal{M}$ to be the set of paths used in a communication scheme, which has one-to-one correspondence to the communication scheme, and denote $\mathcal{A}\subseteq \mathcal{N} \setminus \{a,b\}$ the set of compromised nodes, which uniquely determines Eve's hacking strategy. We introduce a Boolean function $sec(\mathcal{A}, \mathcal{M})$ of a communication scheme and hacking strategy, which is defined as follows.

\begin{definition} \label{def:sec}
For a communication scheme $\mathcal{M}$ and Eve's strategy $\mathcal{A}$, if $\mathcal{M}$ is secure against compromised nodes in $\mathcal{A}$, then $sec(\mathcal{A}, \mathcal{M})=1$. Otherwise, $sec(\mathcal{A}, \mathcal{M})=0$.
\end{definition}

We can see that $sec(\mathcal{A}, \mathcal{M})=0$ if and only if each path in $\mathcal{M}$ goes through at least one nodes in $\mathcal{A}$. In other words, $sec(\mathcal{A}, \mathcal{M})=1$ if and only if there is at least one path in $\mathcal{M}$ which does not go through any nodes in $\mathcal{A}$. We define the security level as follows.
\begin{definition} \label{def:seclevel}
Given a network $G$, the source and sink nodes $a$ and $b$, and the communication scheme $\mathcal{M}$, the security level of $\mathcal{M}$ is the optimal value of the following maximization problem,
\begin{equation}
\begin{aligned}
&\max_{\mathcal{A}} |\mathcal{A}| \\
\mathrm{s.t.} \quad & sec(\mathcal{A}, \mathcal{M})=1.
\end{aligned}
\end{equation}
\end{definition}

We also define a strongest attack to be the most powerful attack that can successfully hack all possible communication schemes.
\begin{definition}
A strongest attack $\mathcal{A}^{st}$ can successfully hack all possible communication schemes,
\begin{equation}\label{Eq:defAst}
\begin{aligned}
\forall \mathcal{M}, \quad sec(\mathcal{A}^{st},\mathcal{M})=0.
\end{aligned}
\end{equation}
\end{definition}

By definition, we see that a strongest attack should contain at least one node of each possible path between Alice and Bob. When Eve compromises a node, we assume that she knows all the keys distributed from (and to) this node. From a security point of view, one can think of Eve making those connecting edges insecure. Given an attack $\mathcal{A}$, define $\mathcal{L}_A\subseteq \mathcal{L}$ as the set of insecure edges caused by this attack. 
If Alice and Bob cannot be connected by a path without using any edges in $\mathcal{L}_A$, no secure path can be found under this attack and such attack is strongest. Thus, we have the following theorem.

\begin{theorem} \label{Thm:StrongCut}
Attack $\mathcal{A}$ is strongest if and only if Alice and Bob belongs to different disjoint subsets partitioned by a cut-set contained in $\mathcal{L}_A$.
\end{theorem}
\emph{Proof.} Proof of ``if": a cut in the graph theory is a partition of the nodes into two disjoint subsets. It determines a cut-set, the set of edges whose two end nodes belongs to different subsets of the partition. Alice and Bob belongs to different subsets. Hence any path connecting Alice and Bob must have at least one edge that connect two nodes of different subset. From the definition, this edge belongs to the cut-set. That is, any path connecting them must contain at least one edge in the cut-set. Then, no secure communication is possible. The attack is strongest.

Proof of ``only if": we need to prove that if $\mathcal{L}_A$ contains no cut-set, there must be a secure path between Alice and Bob. Consider the set of nodes that have secure paths to Alice, if Bob belongs to this set, the proof is done by finding the secure path. If Bob does not belong to it, this set and its compliment set are two disjoint subsets. This partition is a cut. The cut-set must be contained in $\mathcal{L}_A$ and hence $\mathcal{A}$ is strongest.

From the theorem, we can have the following corollary.

\begin{corollary} \label{Cor:noStPath}
If an attack is not strongest, there exists a secure path connecting Alice and Bob.
\end{corollary}


\subsubsection{Communication scheme of the highest security level}
Now, we want to study the most secure communication scheme. That is, such a scheme can tolerate any attacks that other schemes can tolerate. Denote the set of strongest attacks to be $\mathcal{A}^{st}$.

\begin{definition}
A communication scheme, $\mathcal{M}^{h}$, has the highest security level if
\begin{equation}\label{Eq:seclevel}
\begin{aligned}
\mathcal{M}^{h} : sec(\mathcal{A},\mathcal{M}^{h})=\left\{
\begin{aligned}
0 &   & \mathcal{A}\in \{ \mathcal{A}^{st} \} \\
1 &   &  otherwise
\end{aligned}
\right.
\end{aligned}
\end{equation}
\end{definition}

Here we propose a scheme $\mathcal{M}_0$ with the highest security level.
\begin{definition} \label{def:M0}
In the communication scheme $\mathcal{M}_0$, each node in the network except Alice and Bob broadcasts the parity (XOR result of all the keys from the neighbor). Bob receives all the parity information via unencrypted channels (available to Eve). By calculating the parity information, he can share a secret key with Alice.
\end{definition}

We take the network in Figure \ref{fig:strategy1} as an example. Here $c_1$ will announce $k_1 \oplus k_2 \oplus k_4$, $c_2$ will announce $k_2 \oplus k_6 \oplus k_7$, etc. Of course, Alice's and Bob's positions are symmetric. All the parity information can be sent to Alice. The scheme still works.

\begin{theorem}
Scheme $\mathcal{M}_0$, defined in Definition~\ref{def:M0}, is of the highest security level.
\end{theorem}
\emph{Proof.} First, we need to show that this scheme can yield an identical key between Alice and Bob. On Alice's side, she performs the XOR operation to all the keys connected to her and obtains $k_A$. Upon receiving all the parity information from the network, Bob performs XOR operation on all the parity bit string along with his keys connected to his neighbors. Then, all the keys in this network appear in this XOR operation twice except those of the nodes connected directly with Alice. Thus, Bob's XOR result gives $k_A$ and all others are canceled out. In the end, they can achieve an identical key.

Then, we show that the generated key is secure for any attacks that are not strongest. If an attack is not strongest, from Corollary \ref{Cor:noStPath}, we can find a secure path between Alice and Bob. For the scheme $\mathcal{M}_0$, one can think of $k_A$ a secure random key bit string being transmitted from Alice to Bob with one-time pad encryption \cite{vernam1926cipher} and being XOR with some extra random bit strings that might known to Eve. Specifically, suppose the secure path is $a\rightarrow c_1\rightarrow c_2\rightarrow \cdots \rightarrow b$. Then Alice can send her random bit string via this path to Bob. In this case, she adds more unrelated random bit strings, which will not affect the security of the transmission.


Finally, it is obvious that $\mathcal{M}_0$ is insecure under a strongest attack, since it forms a cut between Alice and Bob.

At the end of this part, we have the following remarks. We notice that similar problems have been investigated in both classical \cite{dolev1993perfectly} and quantum networks \cite{salvail2010security}, where the relation between the maximum tolerable compromised nodes and the network connectivity is given. There is also a communication scheme proposed in \cite{beals2008distributed} with probabilistic information-theoretical security. In contrast, our results are independent of the connectivity and not probabilistic, i.e., our communication scheme is secure unless Eve performs the strongest attack. Another remark is that we assume the classical channel is free between any two nodes. Thus, there is no need to consider the efficiency of classical data transmission in the analysis. In practical cases, we need to consider the trade-off between network efficiency and security level. We leave this problem for future works.

\section{Utility optimization and key management} \label{Sc:Key}
When maximizing the security of the network in the previous section, we essentially assume that the key from QKD is sufficient for encryption. While in a practical quantum network, the amount of key is usually limited since QKD is normally far slower than classical communication. In this section, we consider the scenario where the quantum key is a limited source for multiple communication tasks. The problem becomes how to optimize certain network metrics through key management, data scheduling, and routing. For instance, we need to evaluate the encrypted data transmission capacity of a quantum network, i.e., how much data can be transmitted within a unit of time. Here, we borrow techniques in a classical energy harvesting network \cite{Longbo2013Utility}. The main difference is that the key (corresponding to the energy in an energy harvesting network) is defined over channels rather than nodes, which leads to different target functions and constraints in our optimization problem. In this section, we formulate a utility optimization problem to deal with the key management and data routing problem in a QKD network and find an efficient solution based on Lyapunov optimization techniques.

Again, we follow the graph theory expression $G=(\mathcal{N},\mathcal{L})$ to represent the network. Specifically, $a,b\in \mathcal{N}$ represent nodes and $l_{[a,b]}$ represents the link between $a$ and $b$. The time is discretized in the following discussions and $t$ is the index of the time slot. We summarize the notations in Table~\ref{tab:notation} and make the following remarks. The working condition of QKD $S_{[a,c]}(t)$ is a Boolean function and the key management strategy lies in the balance between $S_{[a,c]}(t)K_{[a,c]}$ and $P_{[a,c]}(t)$, representing the key generation and consumption, respectively. During data transmission, we only care about their destinations and classify the data accordingly. For example, we call the data flow with the final destination to node $b$ as type-$b$ data. Data scheduling is determined by $R_a^b(t)$, type-$b$ data admitted to $a$ at time $t$. Since secure data transmission needs encryption, it is given by a function of the key consumption, i.e., $\mu_{[a,c]}(t)=\mu_{[a,c]}(P_{[a,c]}(t))$. In particular, for the case of one-time pad encryption, $\mu_{[a,c]}(t)=P_{[a,c]}(t)$. The total data transmission on an edge $l_{[a,c]}$ is the sum of all types of data transmission, i.e., $\mu_{[a,c]}(t)=\sum_b\mu_{[a,c]}^b(t)$. The key generation rate $K_{[a,c]}$ is determined by the QKD setting between two adjacent nodes, $a$ and $c$. The key is stored in the edge with a storage upper bound $\theta_{[a,c]}$. When the amount of key stored in the edge $l_{[a,c]}\in \mathcal{L}$ is larger than $\theta_{[a,c]}$ at time slot $t$,  QKD in this edge becomes inactive, i.e., $S_{[a,c]}(t)=0$. Compared with the energy harvesting network in literature, we neglect the interaction of different paths due to the optical fiber communications applied in the network.

\begin{table}[htbp]
\centering
\caption{Notations. The subscript $[a,c]$ refers to a quantity defined between the adjacent nodes $a$ and $c$. Denote the data transmitted to $b$ as type-$b$ data.} \label{tab:notation}
\begin{tabular}{c|c}
\hline
Symbol & Interpretation \\
\hline
$N$ & Number of nodes in the network, $|\mathcal{N}|$ \\
$L$ & Number of edges in the network, $|\mathcal{L}|$ \\
$\mathcal{N}_a^{in(out)}$ & Set of nodes connected to (from) node $a$ \\
$Q_a^b(t)$ & Total type-$b$ data queue at node $a$ \\
$E_{[a,c]}(t)$ & Amount of key stored \\
$K_{[a,c]}$ & Amount of key generated per time slot\\
$S_{[a,c]}(t)$ & Working condition of QKD \\
$\theta_{[a,c]}$ & Saturation of key storage \\
$\mu_{[a,c]}(t)$ & Total data transmission \\
$\mu_{[a,c]}^b(t)$ & Type-$b$ data transmission \\
$P_{[a,c]}(t)$ & Key consumption \\
$R_a^b(t)$ & New type-$b$ data transmission request at node $a$ \\
\hline
\end{tabular}
\end{table}

\subsection{Utility optimization problem} \label{sub:utilityopt}
The data transmission capacity problem is a special case of the utility optimization problems. The utility is defined on each data flow, i.e., $U_a^b(R_a^b(t))$, which quantifies how much one can benefit from achieving a data rate $R_a^b(t)$. The concrete expression of the utility function can be defined according to practical applications. A common utility function is concave with the data transmission flow, for example, $U_a^b(R_a^b(t))=k\log_2(R_a^b(t))$, where the coefficient $k$ can be as simple as a constant. 
In particular, when $U_a^b(R_a^b(t))=R_a^b(t)$, the utility optimization problem reduces to the data transmission capacity problem.


The objective of the problem is to optimize the network utility obtained from serving data traffic. Specifically, we consider the following network utility,
\begin{equation}\label{eq:utility}
U_{tot}(\vec{r})=\sum_{a,b\in \mathcal{N} }U_a^b(r_a^b),
\end{equation}
where the average type-$b$ data transmission rate at node $a$ is given by
\begin{equation}\label{Eq:averagerate}
\begin{aligned}
r_a^b \equiv \liminf_{t\rightarrow\infty}\frac{1}{t}\sum_{\tau=0}^{t-1} {R_a^b(\tau)},
\end{aligned}
\end{equation}
and $\vec{r}$ is the matrix with elements of $r_a^b$. In order to evaluate the data transmission capacity for a quantum network, we need to optimize Eq.~\eqref{eq:utility} with certain dynamics and constraints.

\subsection{Dynamics and constraints in a quantum network}
Now, we model the dynamics, shown in Fig.~\ref{fig:dynamics}, and the constraints in the network model. First, we have the key storage dynamics,
\begin{equation}\label{eq:keydynamic}
E_{[a,c]}(t+1)=E_{[a,c]}(t)-P_{[a,c]}(t)+S_{[a,c]}(t)K_{[a,c]},
\end{equation}
where the increase of the key volume $S_{[a,c]}(t)K_{[a,c]}$ comes from QKD and the decrease $-P_{[a,c]}(t)$ is caused by key consumption for encryption. Note that in Eq.~\eqref{eq:keydynamic}, the key storage should be non-negative, $\forall\, t, l_{[a,c]}\in\mathcal{L}$,
\begin{eqnarray}\label{eq:no-underflow}
E_{[a,c]}(t) \geq P_{[a,c]}(t).
\end{eqnarray}
This key availability constraint, Eq.~\eqref{eq:no-underflow}, is a complicated constraint as it couples the key consumption actions across time, i.e., a current $P_{[a,c]}(t)$ decision can affect future actions.

Similarly, we have the data transmission dynamics,
\begin{equation}\label{eq:datadynamic}
Q_a^b(t+1)\leq Q_a^b(t) +R_a^b(t) +\sum_{c\in \mathcal{N}^{in}_a}\mu_{[c,a]}^b(t) -\sum_{c \in \mathcal{N}^{out}_n} \mu_{[a,c]}^b(t).
\end{equation}
The amount of type-$b$ data to be transmitted at node $a$ come from two sources: data flow from other nodes to node $a$, $\sum_{c\in \mathcal{N}^{in}_a}\mu_{[c,a]}^b(t)$; and new data admitted to $a$ for $b$, $R_a^b(t)$. Meanwhile, the queue will decrease if data is transmitted from $a$ to other adjacent nodes $\sum_{c \in \mathcal{N}^{out}_n} \mu_{[a,c]}^b(t)$. The inequality is due to the possibility that neighbor nodes may not have enough data to fulfill the allocated rate. In the following discussions, we just take it as an equality, as when the rate is over-allocated, one can just send some dummy data. Finally, we take account of the stability of the network. That is, the data queue backlog of the whole network needs to be convergent with time,
\begin{equation}\label{eq:stability}
\bar{Q} \equiv  \limsup_{t\rightarrow \infty}\frac{1}{t}\sum_{\tau = 0}^{t-1}\sum_{a,b} {Q_a^b(\tau)}< \infty.
\end{equation}
The stability condition makes sure that all packets admitted into the network are eventually delivered.

\begin{figure*}[hbt]
\centering
\resizebox{16cm}{!}{\includegraphics{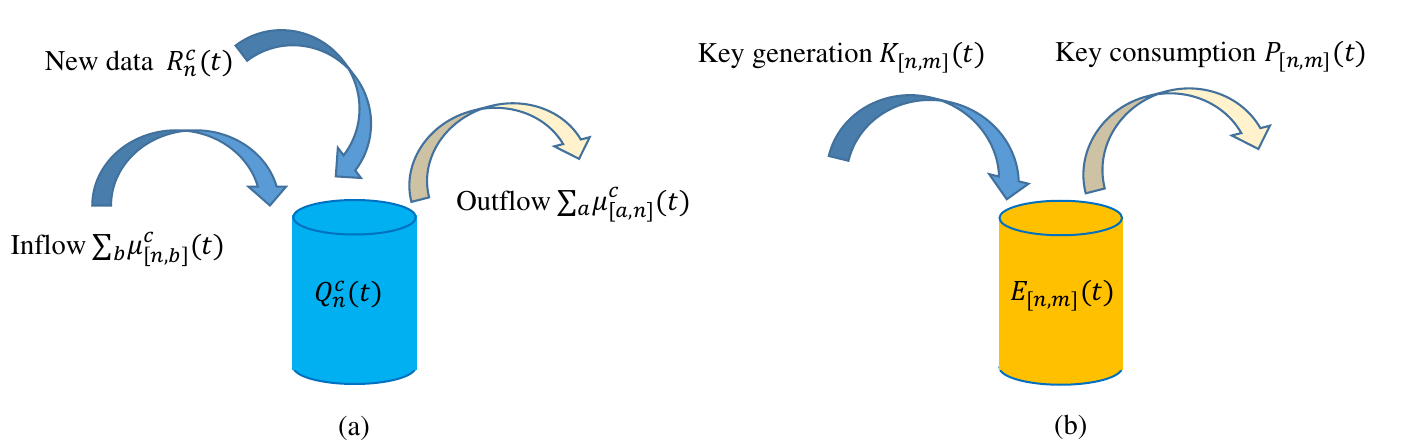}}
\caption{Dynamics in a quantum network. (a) Dynamics of the data queue, as formulated in Eq.~\eqref{eq:datadynamic}. (b) Dynamics of the key storage, as formulated in Eq.~\eqref{eq:keydynamic}.}
\label{fig:dynamics}
\end{figure*}


\subsection{Algorithm design}
To solve the utility optimization problem defined in Section \ref{sub:utilityopt}, we design an algorithm based on the Lyapunov optimization technique \cite{neelynowbook}, which has found wide applications in different network scenarios \cite{net-intelligence-ton,li-ton-15,rahulneely-storage}. Define the Lyapunov function,
\begin{equation}\label{eq:Lyapunov}
L(t) \equiv \frac{1}{2}\sum_{a,b\in \mathcal{N} }[Q_a^b(t)]^2+\frac{1}{2}\sum_{l_{[a,c]}\in \mathcal{L}}[E_{[a,c]}(t)-\theta_{[a,c]}]^2,
\end{equation}
where the storage saturation values $\theta_{[a,c]}$ should be chosen carefully in the algorithm as discussed later in this section. Define the following \emph{drift-plus-penalty} \cite{neelynowbook} for our algorithm design, so as to optimize utility while ensuring network stability,
\begin{equation}\label{eq:target}
\Delta_V(t) \equiv \Delta(t)-V\sum_{a,b\in \mathcal{N} }U(R_a^b(t)),
\end{equation}
where $V$ is a tunable positive constant and
\begin{equation}\label{eq:ldrift}
\Delta(t)=L(t+1)-L(t).
\end{equation}
The construction of the target function, Eq.~\eqref{eq:target}, is similar to the Lagrange multiplier method.

Then, we choose the control action to minimize the drift-plus-penalty given in Eq.~\eqref{eq:target}. Using the queueing dynamics in Eqs.~\eqref{eq:keydynamic} and \eqref{eq:datadynamic}, after some algebras, we decouple the key management and data transmission, so that we can optimize them separately. In the end, the target function Eq.~\eqref{eq:target} can be rewritten as
\begin{equation}\label{eq:targetre}
\begin{aligned}
\Delta_V(t)&\leq B + \sum_{l_{[a,c]}\in \mathcal{L}}(E_{[a,c]}(t)-\theta_{[a,c]})S_{[a,c]}(t)K_{[a,c]} \\
&\quad -\sum_{a,b\in \mathcal{N}}[VU_a^b(R_a^b(t))-Q_a^b(t)R_a^b(t)] \\
& \quad-\sum_{a,b\in \mathcal{N}} \sum_{c\in \mathcal{N}_a^{out}}\mu_{[a,c]}^b(t)[Q_a^b(t)-Q_c^b(t)] \\
& \quad-\sum_{l_{[a,c]}\in \mathcal{L}}(E_{[a,c]}(t)-\theta_{[a,c]})P_{[a,c]}(t).
\end{aligned}
\end{equation}
Here the constant $B$ is given by
\begin{equation}\label{eq:constantB}
\begin{aligned}
B\equiv N^2(\frac{3}{2}d_{max}^2\mu_{max}^2+R_{max}^2)+\frac{L}{2}(P_{max}+K_{max})^2,
\end{aligned}
\end{equation}
where the subscript $max$ means the maximal possible values in the strategies and $d_{max}=\max_a (|\mathcal{N}_a^{in}|,|\mathcal{N}_a^{out}|)$. The detailed derivations of Eq.~\eqref{eq:targetre} is presented in \ref{app:derivation}.

Before we give the utility optimization algorithm, we need to introduce the following network technical terms. The saturation of the key storage, $\theta_{[a,c]}$, is defined as
\begin{equation} \label{eq:deftheta}
\begin{aligned}
\theta_{[a,c]}\equiv \delta \beta V +P_{max},
\end{aligned}
\end{equation}
where $\delta$ is a positive constant satisfying $\mu_{[a,c]}(P_{[a,c]}(t))\leq \delta P_{[a,c]}(t)$ and $\beta$ is the largest first derivative of the utility functions, $\beta= \max_{a,b}\beta_{a,b}=\max_{n,c}(U_a^b)^\prime(0)$. Here, we only consider $(U_a^b)^\prime(0)$ since the utility function is concave. The operational meaning of $\theta_{[a,c]}$ is to let key storage be saturated to a positive constant $\theta_{[a,c]}$ rather than zero since we often need a positive key storage to handle urgent data transmission tasks.

Then, we define the weight of the type-$b$ data over the link $l_{[a,c]}$ as
\begin{equation}\label{eq:linkweight}
  W_{[a,c]}^b(t) = Q_a^b(t)-Q_c^b(t)-\gamma.
\end{equation}
The link weight is given by $W_{[a,c]}(t)=\max_b W_{[a,c]}^b(t)$. Here, $\gamma$ is defined as
\begin{equation}
\begin{aligned}
\gamma \equiv R_{max}+d_{max}\mu_{max},
\end{aligned}
\end{equation}
which means the maximum possible increase of the data queue in a node in a single time slot, including the maximum endogenous increase $d_{max}\mu_{max}$ and exogenous increase $R_{max}$. We consider a data transmission task by some link $l_{[a,c]}$ to be important only when the data queue difference between two nodes, $Q_a^b(t)-Q_c^b(t)$, is large enough (larger than $\gamma$).

The main idea of the algorithm is to optimize the data transmission, $R_a^b(t)$ ($\forall a,b\in \mathcal{N}$), and key management, $P_{[a,c]}(t)$ ($\forall l_{[a,c]}\in \mathcal{L}$), by minimizing the target function, Eq.~\eqref{eq:targetre}, subject to Eqs.~\eqref{eq:no-underflow} and \eqref{eq:datadynamic}. In Eq.~\eqref{eq:targetre} we can see that the optimization of $R_a^b(t)$ and $P_{[a,c]}(t)$ can be done separately.
Note that the network stability constraint Eq.~\eqref{eq:stability} is automatically satisfied under the Lyapunov drift approaches Eq.~\eqref{eq:ldrift}. The total utility Eq.~\eqref{eq:utility} is not optimized directly, but the optimization result can be arbitrarily close to maximum utility of Eq.~\eqref{eq:utility}, which will be discussed in details in Sec.~\ref{sub:performance}.

Now, we present the main optimization algorithm given in Table~\ref{tab:algorithm}, inspired by the energy-limited scheduling algorithm \cite{Longbo2013Utility}.

\begin{table*}
\caption{Utility optimization algorithm.}
\label{tab:algorithm}
\begin{framed}
\begin{enumerate}
\item
\emph{Input of the algorithm}. Initialize $\theta_{[a,c]}$. At every time slot t, observe $Q_a^b(t)$ and $E_{[a,c]}(t)$.
\item
\emph{Key generation}. If $E_{[a,c]}(t)-\theta_{[a,c]}<0$, perform key generation, i.e., let $S_{[a,c]}(t)=1$; otherwise, let $S_{[a,c]}(t)=0$. Note that this decision minimizes the second term on the right-hand side of Eq.~\eqref{eq:targetre}.
\item
\emph{Data transmission}. Make a local optimization on the third term of the right-hand side of Eq.~\eqref{eq:targetre},
\begin{equation}\label{eq:optdata}
\max_{R_a^b(t)}\quad VU_a^b(R_a^b(t))-Q_a^b(t)R_a^b(t),
\end{equation}
with the constraint of $0<R_a^b(t)<R_{max}$.
\item
\emph{Key management}. Optimize the key consumption over all edges, $\vec{P}(t)$, by solving the following maximization
\begin{equation}\label{eq:optpower}
\begin{aligned}
\max_{\vec{P}(t)} G(\vec{P}(t))&= \sum_{n\in \mathcal{N}} \sum_{c\in \mathcal{N}_a^{out}}\mu_{[a,c]}(t)W_{[a,c]}(t) \\&+
\sum_{l_{[a,c]}\in \mathcal{L}}(E_{[a,c]}(t)-\theta_{[a,c]})P_{[a,c]}(t) \\
& = \sum_{l_{[a,c]}\in \mathcal{L}}\left\{\mu_{[a,c]}(t)W_{[a,c]}(t)+ (E_{[a,c]}(t)-\theta_{[a,c]})P_{[a,c]}(t)\right\},
\end{aligned}
\end{equation}
subject to the key availability constraint Eq.~\eqref{eq:no-underflow}. 
\item
\emph{Routing and scheduling}. Find $b^* \in argmax_b W_{[a,c]}^b(t)$. If $W_{[a,c]}^{b^*}(t)>0$, set $\mu_{[a,c]}^{b^*}(t)=\mu_{[a,c]}(t)$, i.e., allocate the full rate over the link $l_{[a,c]}$ to any commodity achieving the maximum positive weight.
\item
\emph{Queue update}. Update $Q_a^b(t)$ and $E_{[a,c]}(t)$ according to their dynamics Eqs.~\eqref{eq:keydynamic} and \eqref{eq:datadynamic}, respectively.
\end{enumerate}
\end{framed}
\end{table*}

\subsection{Analysis of the algorithm and its performance} \label{sub:performance}
Here, we explain how the algorithm works and analyze its performance. We make some remarks on the details of the algorithm. First, the key availability constraint given in Eq.~\eqref{eq:no-underflow} is actually redundant, i.e., we can directly optimize Eq.~\eqref{eq:optpower} without any constraint and obtain the same key management action.
To prove this, we have the following lemma and leave the proof in \ref{app:prooflemma1}

\begin{lemma}\label{lemma:bound}
The data queue and key storage have the following deterministic bounds, $\forall a, b, t, l_{[a,c]}\in \mathcal{L}$,
\begin{equation}
\begin{aligned}
0&\leq Q_a^b(t) \leq \beta V +R_{max}, \\
0&\leq E_{[a,c]}(t) \leq \theta_{[a,c]}+K_{max}.
\end{aligned}
\end{equation}
\end{lemma}

Suppose the optimized key consumption vector obtained by Eq.~\eqref{eq:optpower} is $\vec{P^*}(t)$. Then we consider a new key consumption vector $\vec{P_0}(t)$ by setting $P^*_{[a,c]}(t)$ in $\vec{P^*}(t)$ to be $0$, i.e., the only difference between $\vec{P^*}(t)$ and $\vec{P_0}(t)$ is the key consumption in the link $l_{[a,c]}$. If the constraint Eq.~\eqref{eq:no-underflow} is violated, i.e., $E_{[a,c]}<P_{[a,c]}$, then
\begin{equation}
\begin{aligned}
&G(\vec{P^*}(t))-G(\vec{P_0}(t)) \\
&=\mu_{[a,c]}(P^*_{[a,c]}(t))W_{[a,c]}(t)+\left(E_{[a,c]}(t)-\theta_{[a,c]}\right)P^*_{[a,c]}(t) \\
& \leq \delta P^*_{[a,c]}(t)(\beta V -d_{max}\mu_{max})-\delta\beta V P^*_{[a,c]}(t) \\
& < 0,
\end{aligned}
\end{equation}
which leads to a contradiction that $\vec{P^*}(t)$ is not the optimized strategy. The first inequality is obtained by Lemma~\ref{lemma:bound} and $\mu_{[a,c]}(P^*_{[a,c]}(t))\leq \delta P^*_{[a,c]}(t)$ is due to the definition of $\delta$ of Eq.~\eqref{eq:deftheta}. Especially, for one-time-pad encryption, we have $\mu_{[a,c]}(P^*_{[a,c]}(t))= P^*_{[a,c]}(t)$ and we can take $\delta\geq 1$.

Second, in steps \emph{key management} and \emph{routing and scheduling}, we make an optimization on the destination $b$, i.e., we only consider the destination $b^*$ with the maximum link weight, because
\begin{equation}\label{eq:linkwt}
\begin{aligned}
\sum_b \mu_{[a,c]}^b(t)W_{[a,c]}^b(t) &\leq \sum_b \mu_{[a,c]}^b(t)W_{[a,c]}(t) \\
&=\mu_{[a,c]}(t)W_{[a,c]}(t)
\end{aligned}
\end{equation}
Therefore, it is optimal to allocate the full rate over the link $l_{[a,c]}$ to any commodity achieving the maximum positive weight. If there are multiple destinations $b^*$ achieving the maximum link weight, we can randomly choose one of them to allocate the full rate.

Third, one can see that the optimized target function in the algorithm is different from the original utility function given in Eq.~\eqref{eq:utility}. We want to show that the optimization result of the algorithm can be arbitrary close to the optimal utility $U_{tot}$, i.e., the performance of the algorithm is given by the following theorem and leave its proof in \ref{app:prooftheorem3}.

\begin{theorem}\label{theorem:performance}
The utility optimization result of the algorithm can be arbitrarily close to the optimal utility $U_{tot}$,
\begin{equation} \label{optutility}
\begin{aligned}
\liminf_{\tau\rightarrow \infty}U_{tot}(\vec{r}(\tau)) &=\liminf_{\tau\rightarrow \infty}\sum_{n,c} U_a^b(r_a^b(\tau)) \\
&\geq U_{tot}(\vec{r^*})-\frac{\tilde{B}}{V},
\end{aligned}
\end{equation}
where $r_a^b(\tau)=\frac{1}{\tau}\sum_{t=0}^{\tau-1}R_a^b(t)$ is the average data flow, $\tilde{B}=B+N^2\gamma d_{max}\mu_{max}$ is a constant, and $\vec{r^*}$ is an optimal solution for Eq.~\eqref{eq:utility}.
\end{theorem}

Finally, compared to the original algorithm given in \cite{Longbo2013Utility}, we can optimize the key management in Eq.~\eqref{eq:optpower} locally for each edge rather than each node. This is a particularly useful feature for practical implementation.

\subsection{Simulation}
To test our algorithm, we make a simulation of the toy network model given in Fig.~\ref{fig:strategy1}. Here we consider a simple task where the data are transmitted from Alice to Bob, i.e., $R_a^b(t)=0$ for all nodes $a$ except Alice's node. From the network structure in Fig.~\ref{fig:strategy1} we can see that $d_{max}=2$. We set the values of other parameters as $\beta=1$, $\delta=2$, $P_{max}=\mu_{max}=2$, $R_{max}=3$, $K_{[a,c]}=0.1, \forall l_{[a,c]}$, $\gamma = d_{max}\mu_{max}+R_{max}=7$ and $\theta_{[a,c]} = \delta \beta V +P_{max} =2V+2$. These values are just taken for simplicity and we do not consider their units. The utility function is taken as $U(r) = \ln(1+r)$. The optimization of such a utility function is actually also a maximization of data transmission. We simulate the utility of the whole network versus $V\in \{5,10,15,20,25,30,35,40,45\}$ to see convergence behavior when $V$ goes larger.

The simulation of the total utility is given in Fig.~\ref{fig:utility}. It turns out that the utility converges quickly to the optimal value of 0.1815, which shows our algorithm is quite efficient. We also show the evolution of the total data queue and key storage in Fig.~\ref{fig:QandEvsT} with $V=40$. We can observe that the network become stable after $t=10^3$. and the data queue grows slowly during $t\in[10^2, 10^3]$, which is because the amount of data queue exceeds the key storage in the previous time period. After the stability of the network, the gap between the key storage and data queue is determined by the parameter settings, which can be optimized and is left for future works. The evolution of the utility is shown in Fig.~\ref{fig:UvsT}. The drop at the beginning comes from the initialization of the network. The data transimission at the beginning is high since the data and key storage is initialized to be zero, which leads to a high utility. The convergence speed is similar to the original algorithm \cite{Longbo2013Utility}.

In our simulation of the toy model, we choose typical values for the parameters $\beta$, $\delta$, $P_{max}$, $\mu_{max}$, $R_{max}$, and $K_{[a,c]}$ and assume that the key rates in different links are the same, $K_{[a,c]}=0.1, \forall l_{[a,c]}$. While in a practical field test we can substitute some real values into the simulation, such as the QKD key rates in each link and the classical channel capacity. We can obtain useful results with the real values, for example, the convergence time of data queue and key storage, the saturation of the data queue and key storage, and the data transmission capacity of the network.

We make comparisons with some other routing protocols in the literature. Some protocols are proposed for different targets in a QKD network. For example, in ref.~\cite{amer2020efficient}, the authors apply a multi-path routing scheme to maximize the key rate between two remote end-users in a network with both quantum repeaters and trusted nodes. While in ref.~\cite{tanizawa2016routing}, the authors consider a routing protocol based on current key storage. It applies a modified Open Shortest Path Fast routing algorithm to minimize the path length between end-users and reduce the key consumption in the meantime. There are also some routing protocols in energy harvesting networks, aiming at optimizing the utility function of the data flow. In \cite{chen2011finite}, the utility optimization has been formulated as a standard convex optimization problem without Lyapunov optimization techniques, however, it requires future knowledge on the energy harvesting that is hardly available. Another protocol based on a convex optimization problem without Lyapunov optimization techniques is proposed in ref.~\cite{chen2013simple}, which gives rise to an asymptotic optimal solution. By simulating the simplest network with one source and one destination, it shows a better performance on utility compared with the optimizations based on Lyapunov optimization techniques.

\begin{figure}[hbt]
\centering
\resizebox{8cm}{!}{\includegraphics{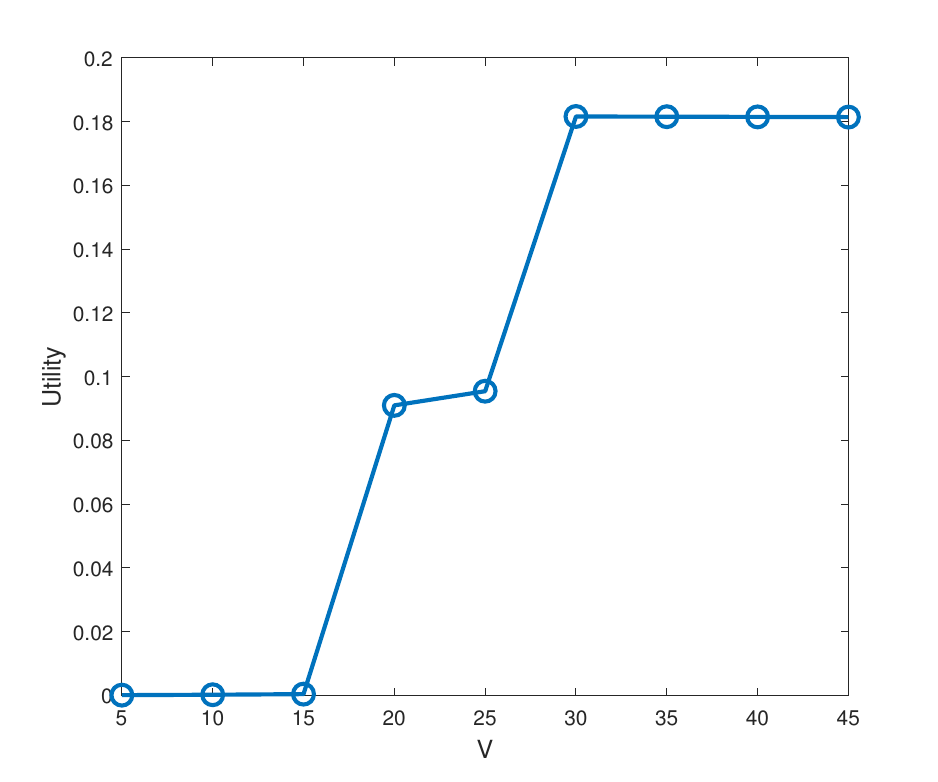}}
\caption{Simulation result of a toy network model. The optimized utility (also maximized data transmission) converges quickly to the optimal value.}
\label{fig:utility}
\end{figure}

\begin{figure}[hbt]
\centering
\resizebox{8cm}{!}{\includegraphics{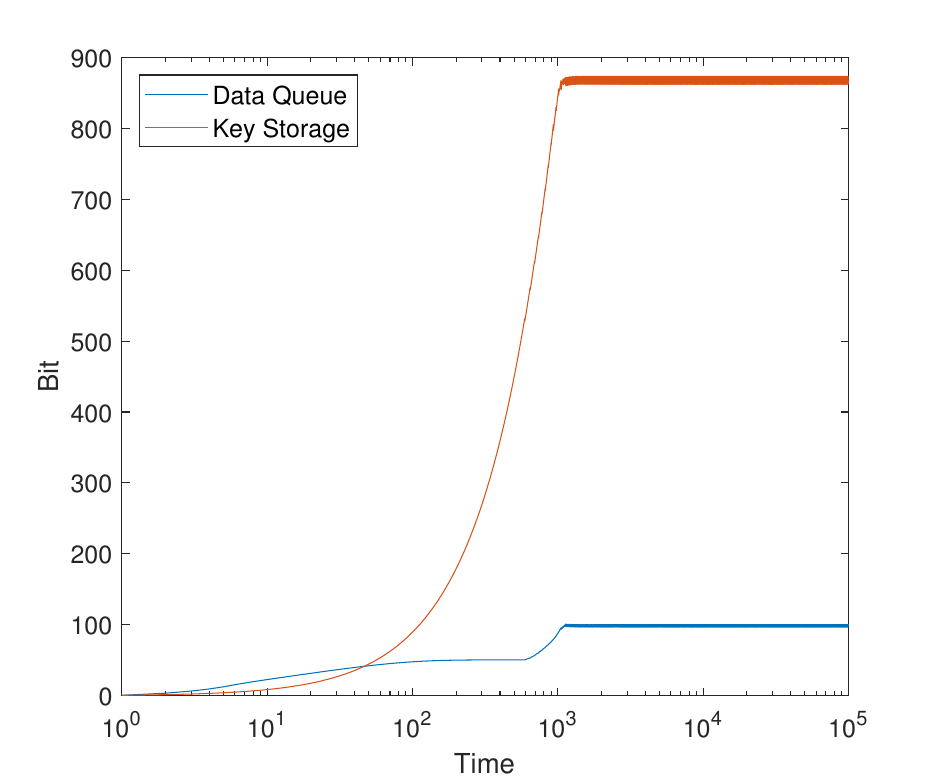}}
\caption{Evolution of the total data queue and key storage of the toy network with $V=40$. We can see that the network become stable after $t=10^3$. The width of both curves shows the data fluctuations.}
\label{fig:QandEvsT}
\end{figure}

\begin{figure}[hbt]
\centering
\resizebox{8cm}{!}{\includegraphics{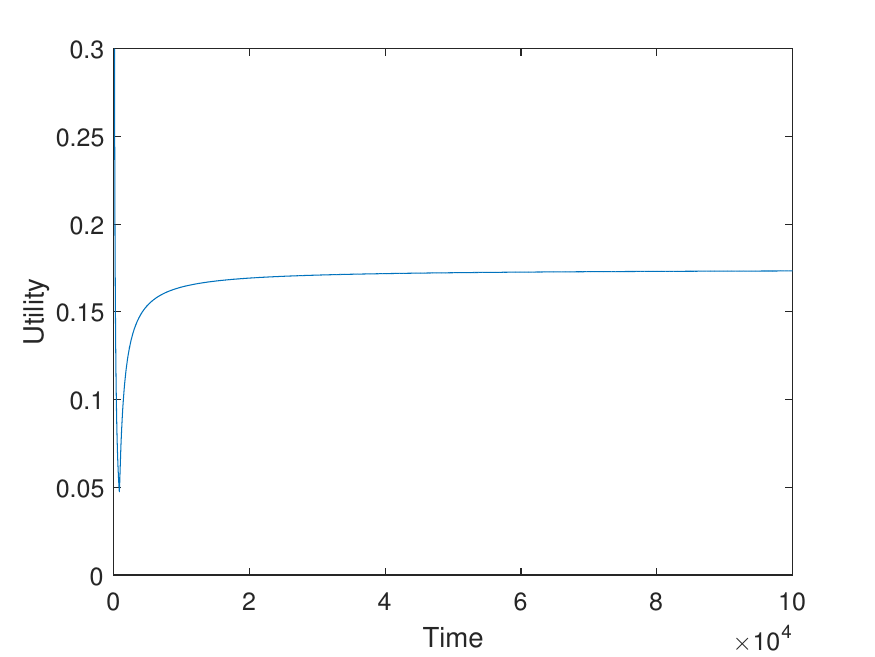}}
\caption{Evolution of the utility of the toy network with $V=40$.}
\label{fig:UvsT}
\end{figure}

\section{Discussion and conclusion}\label{sec:conclusion}

In this work, we propose solutions to two typical and crucial issues in quantum networks, namely security and key management. We tackle the security issue with graph theory and design a communication scheme of the highest security level, where each intermediate node broadcasts the XOR result of all its keys. To optimize the utility of the data. These two problems are closely connected in some special situations. Suppose the key is free and data cannot be stored in each node. Then our optimization problem will reduce to the maximum flow problem in a directed graph. If we further assume the capacities of the classical links are the same, this maximum flow will be proportional to the security level according to the max-flow min-cut theorem.

In this paper, we consider two networks, a classical network for data transmission and a QKD network for key generation. The latter can be naturally extended to an entanglement distribution network in the future, where the key distribution and XOR operation correspond to the Einstein–Podolsky–Rosen (EPR) pair distribution and entanglement swapping. It is interesting to apply the techniques used in this work to an entanglement distribution network.

For future works, one can substitute the data communication requests and key rate of an actual quantum network (such as the Hefei 46-node network) to our simulations and make a field test of our results. One can also consider more complex topological structures and other practical issues such as latency and scalability.

Finally, a trusted node does not need to perform full QKD with users, i.e., the privacy amplification process can be omitted first and raw keys can be directly exchanged \cite{Stacey2015Relay}. We call such a node an honest but curious node. In this case, the intermediate nodes lie between trusted and untrusted. The security assessment needs more complicated analysis.

\appendix

\section{Derivations of Eq.~\eqref{eq:targetre}}\label{app:derivation}
By the definition of $\Delta(t)$ in Eq.~\eqref{eq:ldrift}, we divide $\Delta(t)$ into two parts. The first part comes from the data queue term,
\begin{equation}\label{appeq:diff}
\begin{aligned}
&\frac{1}{2}\sum_{a,b\in \mathcal{N}}[Q_a^b(t+1)]^2-\frac{1}{2}[Q_a^b(t)]^2 \\
&= \sum_{a,b\in \mathcal{N}}Q_a^b(t)\left(-\sum_{c\in \mathcal{N}_a^{out}}\mu^b_{[a,c]}(t)+\sum_{c\in \mathcal{N}_a^{in}} \mu^b_{[c,a]}(t)+ R_a^b(t)\right) \\
& +\frac{1}{2}\left(-\sum_{c\in \mathcal{N}_a^{out}}\mu^b_{[a,c]}(t)+\sum_{c\in \mathcal{N}_a^{in}} \mu^b_{[c,a]}(t)+ R_a^b(t)\right)^2.
\end{aligned}
\end{equation}
For the first term in the rhs. of Eq.~\eqref{appeq:diff}, we want to show that
\begin{equation}\label{appeq:queue}
\begin{aligned}
&\sum_{a,b}Q_a^b(t)\left(-\sum_{c\in \mathcal{N}_a^{out}}\mu^b_{[a,c]}(t)+\sum_{c\in \mathcal{N}_a^{in}} \mu^b_{[c,a]}(t)\right) \\
= &- \sum_{a,b\in \mathcal{N}} \sum_{c\in \mathcal{N}_a^{out}} \mu_{[a,c]}^b(t)[Q_a^b(t)-Q_c^b(t)].
\end{aligned}
\end{equation}
Consider an arbitrary term, $Q_a^b(t)\left(-\mu^b_{[a,c]}(t)+\mu^{b}_{[a^\prime,a]}(t)\right)$, there will always be another term in the summation $Q_c^b(t)\left(-\mu^b_{[c,c^\prime]}(t)+ \mu^b_{[a,c]}(t)\right)$. We can regroup these terms and obtain $-\mu^b_{[a,c]}(t)[Q_a^b(t)-Q_c^b(t)]$. Similarly, we can do this for all other terms in the summation and get the rhs. of Eq.~\eqref{appeq:queue}.

For the second term in the rhs. of Eq.~\eqref{appeq:diff}, we have
\begin{equation}
\begin{aligned}
&\left(-\sum_{c\in \mathcal{N}_a^{out}}\mu^b_{[a,c]}(t)+\sum_{c\in \mathcal{N}_a^{in}} \mu^b_{[c,a]}(t)+ R_a^b(t)\right)^2 \\
&\leq \left(\sum_{c\in \mathcal{N}_a^{out}}\mu^b_{[a,c]}(t)\right)^2+\left(\sum_{c\in \mathcal{N}_a^{in}} \mu^b_{[c,a]}(t)+ R_a^b(t)\right)^2 \\
&\leq d^2_{max}\mu^2_{max}+ (d_{max}\mu_{max}+R_{max})^2 \\
&\leq d^2_{max}\mu^2_{max}+ \frac{1}{2}(d^2_{max}\mu^2_{max}+R^2_{max}) \\
&=3d^2_{max}\mu^2_{max}+2R^2_{max}.
\end{aligned}
\end{equation}

Similar calculations can be done for the second part of $\Delta(t)$ which comes from the key storage term. Finally we can get Eq.~\eqref{eq:targetre} by some algebras.

\section{Proof of Lemma~\ref{lemma:bound}}\label{app:prooflemma1}
We prove this lemma with mathematical induction. First we can easily see that the bound holds for $t=0$, since $Q_a^b(0)=0$ and $E_{[a,c]}(0)=0$.

Then we prove that if $0\leq Q_a^b(t) \leq \beta V +R_{max}$, then $0\leq Q_a^b(t+1) \leq \beta V +R_{max}$. According to the dynamics of data queue Eq.~\eqref{eq:datadynamic}, we can see that the increase of the data queue comes from two aspects: endogenous data $\sum_{c\in \mathcal{N}_a^{in}}\mu_{[c,a]}^b(t)$ and exogenous data $R_a^b(t)$. We consider the following two exclusive cases: first, if there are endogenous data, then they must come from at least one other node, say $c$. From Eq.~\eqref{eq:linkweight}, if there is a data flow from $c$ to $a$ at time $t$, their data queues at time $t$ must satisfy,
\begin{equation}
Q_a^b(t)\leq Q_c^b(t)-\gamma\leq \beta V+R_{max}-\gamma.
\end{equation}
From time slot $t$ to $t+1$, the maximum possible data queue increase of one node is $\gamma$. Then we have $Q_a^b(t+1)\leq  \beta V+R_{max}$; second, there are no endogenous data, which means there are only exogenous data or there are no data queue increase at all. From Eq.~\eqref{eq:optdata}, the existence of a valid optimization result $R_a^b(t)$ requires $Q_a^b(t)\leq\beta_{a,b} V\leq\beta V$. From time slot $t$ to $t+1$, the maximum possible exogenous data queue increase is $R_{max}$. Then we also have $Q_a^b(t+1)\leq  \beta V+R_{max}$. If there are no data queue increase from $t$ to $t+1$, it is straightforward that $Q_a^b(t+1)\leq Q_a^b(t)\leq \beta V+R_{max}$.

Finally we prove if $0\leq E_{[a,c]}(t)\leq \theta_{[a,c]}+K_{max}$, then $0\leq E_{[a,c]}(t+1)\leq \theta_{[a,c]}+K_{max}$. We also consider the following two exclusive cases: according to the second step of the algorithm, if $E_{[a,c]}(t)<\theta_{[a,c]}$, there will be key generation with a maximum key of $K_{max}$, then $E_{[a,c]}(t+1)<\theta_{[a,c]}+K_{max}$; if $E_{[a,c]}(t)\geq \theta_{[a,c]}$, there will be no key generation and $E_{[a,c]}(t+1)\leq E_{[a,c]}(t)\leq \theta_{[a,c]}+K_{max}$.
\section{Proof of Theorem~\ref{theorem:performance}}\label{app:prooftheorem3}

In the algorithm we minimize the following function at time $t$
\begin{equation}\label{eq:policy1}
\begin{aligned}
D(t)&= \sum_{l_{[a,c]}\in \mathcal{L}}(E_{[a,c]}(t)-\theta_{[a,c]})S_{[a,c]}(t)K_{[a,c]} \\
& -\sum_{a,b\in \mathcal{N}}[VU_a^b(R_a^b(t))-Q_a^b(t)R_a^b(t)] \\
& -\sum_{a,b\in \mathcal{N}} \sum_{c\in \mathcal{N}_a^{out}}\mu_{[a,c]}^b(t)[Q_a^b(t)-Q_c^b(t)-\gamma] \\
& -\sum_{l_{[a,c]}\in \mathcal{L}}(E_{[a,c]}(t)-\theta_{[a,c]})P_{[a,c]}(t),
\end{aligned}
\end{equation}
and get an optimized set of strategies $\mathcal{S}=\{R_a^b(t),P_{[a,c]}(t)\}$. Now we consider another function
\begin{equation}\label{eq:policy2}
\begin{aligned}
\tilde{D}(t)&= \sum_{l_{[a,c]}\in \mathcal{L}}(E_{[a,c]}(t)-\theta_{[a,c]})S_{[a,c]}(t)K_{[a,c]} \\
& -\sum_{a,b\in \mathcal{N}}[VU_a^b(R_a^b(t))-Q_a^b(t)R_a^b(t)] \\
& -\sum_{a,b\in \mathcal{N}} \sum_{c\in \mathcal{N}_a^{out}}\mu_{[a,c]}^b(t)[Q_a^b(t)-Q_b^c(t)] \\
& -\sum_{l_{[a,c]}\in \mathcal{L}}(E_{[a,c]}(t)-\theta_{[a,c]})P_{[a,c]}(t) \\
& = D(t) - \sum_{a,b\in \mathcal{N}} \sum_{c\in \mathcal{N}_a^{out}}\mu_{[a,c]}^b(t) \gamma.
\end{aligned}
\end{equation}
We can see that the only difference between Eq.~\eqref{eq:policy1} and Eq.~\eqref{eq:policy2} is that in Eq.~\eqref{eq:policy2} $\gamma$ is not introduced. Suppose the optimized strategy to minimize Eq.~\eqref{eq:policy2} is $\tilde{\mathcal{S}}=\{\tilde{R}_a^b(t),\tilde{P}_{[a,c]}(t)\}$. We have the following relation,
\begin{equation}\label{eq:dfunctionrelation}
\begin{aligned}
D^\mathcal{S}(t)&\leq D^{\tilde{\mathcal{S}}}(t) \\
\tilde{D}^{\tilde{\mathcal{S}}} (t) &\leq \tilde{D}^\mathcal{S}(t).
\end{aligned}
\end{equation}
From the first inequality of Eq.~\eqref{eq:dfunctionrelation} we quickly get
\begin{equation}\label{eq:policyrelation}
\tilde{D}^\mathcal{S}(t)\leq \tilde{D}^{\tilde{\mathcal{S}}}(t)+N^2 \gamma d_{max} \mu_{max}.
\end{equation}
Then we substitute Eq.~\eqref{eq:policyrelation} to Eq.~\eqref{eq:targetre},
\begin{equation}
\Delta(t)-V\sum_{a,b\in \mathcal{N}}U_a^b(R_a^b(t))\leq B + \tilde{D}^\mathcal{S}(t)\leq \tilde{B}+ \tilde{D}^{\tilde{\mathcal{S}}} (t),
\end{equation}
where $\tilde{B}=B+N^2 \gamma d_{max} \mu_{max}$. Since the strategy $\tilde{\mathcal{S}}=\{\tilde{R}_a^b(t),\tilde{P}_{[a,c]}(t)\}$ can take continuous values, we apply the relation $-\tilde{D}^{\tilde{\mathcal{S}}}(t)\geq VU_{tot}(\vec{r^*})$ given in Theorem 1 of \cite{Longbo2013Utility} and Claim 1 of \cite{neely2006energy}. This means the maximization of the total utility in a single time slot will be larger than that in time average,
\begin{equation}\label{eq:deltat}
\Delta(t)-V\sum_{a,b\in \mathcal{N}}U_a^b(R_a^b(t))\leq \tilde{B}- VU_{tot}(\vec{r^*}).
\end{equation}
Then we sum over Eq.~\eqref{eq:deltat} from $t=0$ to $t=\tau-1$,
\begin{equation}
L(\tau)-L(0)-V\sum_{t=0}^{\tau-1} \sum_{a,b\in \mathcal{N}}U_a^b(R_a^b(t)) \leq \tau\tilde{B}- \tau VU_{tot}(\vec{r^*}).
\end{equation}
Divide $V\tau$ in both sides and use $L(\tau)\geq 0$ and $L(0)=0$,
\begin{equation}
\begin{aligned}
\sum_{a,b\in \mathcal{N}} U_a^b(\frac{1}{\tau}\sum_{t=0}^{\tau-1}R_a^b(t))&\geq \frac{1}{\tau}\sum_{t=0}^{\tau-1} \sum_{a,b\in \mathcal{N}}U_a^b(R_a^b(t)) \\
&\geq  U_{tot}(\vec{r^*})-\frac{\tilde{B}}{V},
\end{aligned}
\end{equation}
where the first inequality is because the utility function is concave. Then we take a $\mathrm{\liminf}$ as $\tau \rightarrow \infty$ and have
\begin{equation}
\liminf_{\tau \rightarrow \infty}\sum_{a,b} U_a^b(r_a^b(\tau ))\geq U_{tot}(\vec{r^*})-\frac{\tilde{B}}{V}.
\end{equation}

\acknowledgments
The authors would like to thank T.~Chen, Y.~Liu, and S.~Yao for enlightening discussions. This work is supported by the National Natural Science Foundation of China Grants No.~11875173 and No.~11674193 and the National Key R\&D Program of China Grants No.~2017YFA0303900 and No.~2017YFA0304004.

\bibliographystyle{IEEEtran}

\bibliography{bibqkdnetwork}

\end{document}